\documentclass[12pt]{article}%
\usepackage{amsmath}
\usepackage{amsfonts}
\usepackage{amssymb}
\usepackage{graphicx}%
\newtheorem{theorem}{Theorem}

\newtheorem{definition}[theorem]{Definition}
\newtheorem{example}[theorem]{Example}

\newtheorem{proposition}[theorem]{Proposition}
\newtheorem{remark}[theorem]{Remark}

\newenvironment{proof}[1][Proof]{\noindent\textbf{#1.} }{\ \rule{0.5em}{0.5em}}
\begin{document}

\title{Construction of a transmutation for the one-dimensional Schr\"{o}dinger
operator and a representation for solutions}
\author{Vladislav V. Kravchenko\\{\small Departamento de Matem\'{a}ticas, CINVESTAV del IPN, Unidad
Quer\'{e}taro, }\\{\small Libramiento Norponiente No. 2000, Fracc. Real de Juriquilla,
Quer\'{e}taro, Qro. C.P. 76230 MEXICO}\\{\small e-mail: vkravchenko@math.cinvestav.edu.mx \thanks{Research was
supported by CONACYT, Mexico via the project 166141.}}}
\maketitle

\begin{abstract}
A new representation for solutions of the one-dimensional Schr\"{o}dinger
equation $-u^{\prime\prime}+q(x)u=\omega^{2}u$ is obtained in the form of a
series possessing the following attractive feature. The truncation error is
$\omega$-independent for all $\omega\in\mathbb{R}$. For the coefficients of
the series simple recurrent integration formulas are obtained which make the
new representation applicable for computation.

\end{abstract}

\section{Introduction}

In the present work a new representation for solutions of the one-dimensional
Schr\"{o}dinger equation%
\begin{equation}
-u^{\prime\prime}+q(x)u=\omega^{2}u \label{Intro Schr}%
\end{equation}
is obtained in the form of a series possessing the following attractive
feature. If $u_{N}(\omega,x)$ denotes the truncated series and $u(\omega,x)$
the exact solution, the following inequality is valid for any $\omega
\in\mathbb{R}$, $\left\vert u(\omega,x)-u_{N}(\omega,x)\right\vert
\leq\varepsilon_{N}\left(  x\right)  $, where $\varepsilon_{N}$ is a
nonnegative function independent of $\omega$ and $\varepsilon_{N}\left(
x\right)  \rightarrow0$ when $N\rightarrow\infty$. For the coefficients of the
series simple recurrent integration formulas are obtained which make the new
representation applicable for computation.

In the recent work \cite{KNT} a representation for solutions of
(\ref{Intro Schr}) possessing the described above feature was proposed in a
completely different form. Both representations are united by the fact that
they are obtained with the use of a transmutation (transformation) operator.
In \cite{KNT} the solution $u$ of (\ref{Intro Schr}) satisfying the initial
conditions
\[
u(0)=1,\qquad u^{\prime}(0)=-i\omega
\]
was considered in the form%
\[
u(\omega,x)=e^{-i\omega x}+\int_{-x}^{x}K(x,y)e^{-i\omega y}dy,
\]
well known from \cite{Marchenko}, \cite{Marchenko52} and numerous other
publications. The kernel $K$ was found in \cite{KNT} in the form of a
Fourier-Legendre series which led to a representation of $u(\omega,x)$ in the
form of a Neumann series of Bessel functions (NSBF). In the present work we
explore another possibility of representing $u(\omega,x)$ as a result of
action of a transmutation operator. Namely, an elementary reasoning (see the
next section) based on the well known facts from the scattering theory leads
to another representation of $u(\omega,x)$ in the form
\begin{equation}
u(\omega,x)=e^{-i\omega x}+\int_{-\infty}^{x}\mathbf{A}(x,y)e^{-i\omega y}dy
\label{Intro u}%
\end{equation}
where the kernel $\mathbf{A}(x,\cdot)\in L_{2}\left(  -\infty,x\right)  $ is
that arising in the scattering theory associated with (\ref{Intro Schr}), see,
e.g., \cite{Faddeev}. Since the integral in (\ref{Intro u}) is taken over a
semi-infinite interval it is natural to look for $\mathbf{A}(x,\cdot)$ in the
form of a Fourier-Laguerre series. This is done in the present work, and as a
corollary a new representation for solutions of (\ref{Intro Schr}) is
obtained. For $\omega$-independent coefficients of the representation a direct
formula is derived in terms of so-called formal powers arising in spectral
parameter power series (SPPS) method \cite{KKRosu}, \cite{KrPorter2010},
\cite{KT Birkhauser 2013}. Moreover, a much more convenient for computing
recurrent integration formula is obtained as well, which in practice allows
one to compute thousands of the coefficients. We illustrate several features
of the new representation numerically on a simple test problem. All the
computations reported here took not more than several seconds performed in
Matlab 2012, including those which involved computation of up to $10^{5}$
coefficients. Although this paper is not about a new numerical algorithm, and
we do not discuss details of its implementation, our numerical results show
that the new representation can be of interest for practical computation.

Finally, the new representation can be extended onto a general Sturm-Liouville
equation as well as onto the perturbed Bessel equation (see \cite{KT PQR} and
\cite{KTC} where it was done for the NSBF representation).

Besides this Introduction the paper contains four sections. In Section 2 we
obtain the Fourier-Laguerre expansion of the kernel $\mathbf{A}$. In Section 3
we prove the main result of this work, the new representation for solutions of
(\ref{Intro Schr}). In Section 4 we derive a recurrent integration procedure
for computing the coefficients in the representation. Section 5 contains some
numerical illustrations.

\section{Construction of a transmutation}

Consider the equation
\begin{equation}
-u^{\prime\prime}+q(x)u=\omega^{2}u \label{Schr}%
\end{equation}
on a finite interval $(0,d)$. We suppose that $q$ is a real valued, measurable
function. The solution of (\ref{Schr}) satisfying the initial conditions%
\begin{equation}
u(0)=1,\qquad u^{\prime}(0)=-i\omega, \label{init u}%
\end{equation}
will be denoted as $u(\omega,x)$. We consider $\omega\in\mathbb{R}$.

One can extend $q$ by zero onto the whole line, and $u(\omega,x)$ by
$e^{-i\omega x}$ onto the half-line $(-\infty,0)$. Then $u(\omega,x)$ can be
regarded as a Jost solution of (\ref{Schr}) satisfying the asymptotic relation
(which is in fact an equality in our case) $u(\omega,x)\sim e^{-i\omega x}$
when $x\rightarrow-\infty$. Hence, it is known (see, e.g., \cite{Faddeev})
that there exists such a function $\mathbf{A}(x,\cdot)\in L_{2}\left(
-\infty,x\right)  $ that
\begin{equation}
u(\omega,x)=e^{-i\omega x}+\int_{-\infty}^{x}\mathbf{A}(x,y)e^{-i\omega
y}dy\qquad\text{for all }\omega\text{.} \label{u=Ae}%
\end{equation}
This integral representation of $u(\omega,x)$ can be viewed as action of an
operator of transmutation (transformation) with the kernel $\mathbf{A}(x,y)$
on the solution $e^{-i\omega x}$ of the elementary equation $-u^{\prime\prime
}=\omega^{2}u$. Denote this operator by $A\left[  v\right]  (x):=v(x)+\int
_{-\infty}^{x}\mathbf{A}(x,y)v(y)dy$.

Note that (see, e.g., \cite{Faddeev})
\begin{equation}
\mathbf{A}(x,x)=\frac{1}{2}\int_{0}^{x}q(y)dy. \label{A(x,x)}%
\end{equation}

By a change of the integration variable equality (\ref{u=Ae}) can be written
as follows
\[
u(\omega,x)=e^{-i\omega x}\left(  1+\int_{0}^{\infty}\mathbf{A}%
(x,x-t)e^{i\omega t}dt\right)  \text{.}%
\]
Let us represent the kernel $\mathbf{A}(x,x-t)$ in the form $\mathbf{A}%
(x,x-t)=\mathbf{a}(x,t)e^{-t}$. The function $\mathbf{a}(x,\cdot)$ then
belongs to the space $L_{2}\left(  0,\infty;e^{-t}\right)  $ equipped with the
scalar product $\left\langle u,v\right\rangle :=\int_{0}^{\infty}%
u(t)\overline{v}(t)e^{-t}dt$ and the norm $\left\Vert u\right\Vert
:=\sqrt{\left\langle u,u\right\rangle }$. Thus, for any $x\in\left[
0,d\right]  $ the function $\mathbf{a}(x,\cdot)$ admits a Fourier-Laguerre
expansion convergent in this norm,
\[
\mathbf{a}(x,t)=\sum_{n=0}^{\infty}a_{n}(x)L_{n}(t),
\]
where $L_{n}$ stands for the Laguerre polynomial of order $n,$ and hence%
\begin{equation}
\mathbf{A}(x,y)=\sum_{n=0}^{\infty}a_{n}(x)L_{n}(x-y)e^{-(x-y)}. \label{A}%
\end{equation}
We note that due to (\ref{A(x,x)}) and the fact that $L_{n}(0)=1$ we have the
equality
\begin{equation}
\sum_{n=0}^{\infty}a_{n}(x)=\frac{1}{2}\int_{0}^{x}q(y)dy. \label{Sum a_n}%
\end{equation}
In order to find formulas for the coefficients $a_{n}$ we introduce first the
following notations.

Throughout the paper we suppose that $f_{0}$ is a solution of the equation
\begin{equation}
f^{\prime\prime}-q(x)f=0 \label{SchrHom}%
\end{equation}
satisfying the initial conditions
\[
f_{0}(0)=1,\quad f_{0}^{\prime}(0)=0.
\]

Consider two sequences of recursive integrals (see \cite{KrCV08},
\cite{KrPorter2010})
\begin{equation}
X^{(0)}(x)\equiv1,\qquad X^{(n)}(x)=n\int_{0}^{x}X^{(n-1)}(s)\left(  f_{0}%
^{2}(s)\right)  ^{(-1)^{n}}\,\mathrm{d}s,\qquad n=1,2,\ldots\label{Xn}%
\end{equation}
and
\begin{equation}
\widetilde{X}^{(0)}\equiv1,\qquad\widetilde{X}^{(n)}(x)=n\int_{0}%
^{x}\widetilde{X}^{(n-1)}(s)\left(  f_{0}^{2}(s)\right)  ^{(-1)^{n-1}%
}\,\mathrm{d}s,\qquad n=1,2,\ldots. \label{Xtilde}%
\end{equation}

\begin{definition}
\label{Def Formal powers phik and psik}The family of functions $\left\{
\varphi_{k}\right\}  _{k=0}^{\infty}$ constructed according to the rule
\begin{equation}
\varphi_{k}(x)=%
\begin{cases}
f_{0}(x)X^{(k)}(x), & k\text{\ odd},\\
f_{0}(x)\widetilde{X}^{(k)}(x), & k\text{\ even}%
\end{cases}
\label{phik}%
\end{equation}
is called the system of formal powers associated with $f_{0}$.
\end{definition}

\begin{remark}
\label{Rem On particular solutions}If $f_{0}$ has zeros some of the recurrent
integrals (\ref{Xn}) or (\ref{Xtilde}) may not exist, although even in that
case the formal powers (\ref{phik}) are well defined. It is convenient to
construct them in the following way. Take a nonvanishing solution $f$ of
(\ref{Schr}) such that $f(0)=1$. For example, $f=f_{0}+if_{1}$ where $f_{1}$
is a solution of (\ref{Schr}) satisfying $f_{1}(0)=0$, $f_{1}^{\prime}(0)=1$.
Since $q$ is real valued such $f$ does not vanish. Then (see \cite[Proposition
4.7]{KT Birkhauser 2013})
\[
\varphi_{k}=%
\begin{cases}
\Phi_{k}, & k\text{\ odd,}\\
\Phi_{k}-\frac{f^{\prime}(0)}{k+1}\Phi_{k+1}, & k\text{\ even,}%
\end{cases}
\]
where $\Phi_{k}$ are formal powers associated with $f$.
\end{remark}

\begin{remark}
The formal powers arise in the spectral parameter power series (SPPS)
representation for solutions of (\ref{Schr}) (see \cite{KKRosu},
\cite{KrCV08}, \cite{KMoT}, \cite{KrPorter2010}). In particular, the solution
$u(\omega,x)$ has the form
\begin{equation}
u(\omega,x)=\sum_{n=0}^{\infty}\frac{\left(  -i\omega\right)  ^{n}\varphi
_{n}(x)}{n!}.\label{u SPPS}%
\end{equation}
The series converges uniformly both with respect to $x$ on $\left[
0,d\right]  $ and with respect to $\omega$ on any compact subset of the
complex plane.
\end{remark}

$\bigskip$

\begin{proposition}%
\begin{equation}
A\left[  x^{k}\right]  =\varphi_{k}. \label{Axk}%
\end{equation}

\end{proposition}

\begin{proof}
From (\ref{u=Ae}) and (\ref{u SPPS}) the following equality follows%
\[
\sum_{n=0}^{\infty}\frac{\left(  -i\omega\right)  ^{n}\varphi_{n}(x)}{n!}%
=\sum_{n=0}^{\infty}\frac{\left(  -i\omega\right)  ^{n}x^{n}}{n!}+\sum
_{n=0}^{\infty}\frac{\left(  -i\omega\right)  ^{n}}{n!}\int_{-\infty}%
^{x}\mathbf{A}(x,y)y^{n}dy\qquad\text{for all }\omega
\]
from which (\ref{Axk}) is obtained by equating expressions at corresponding
powers of $\omega$.
\end{proof}

\begin{proposition}
The coefficients $a_{n}$ in (\ref{A}) have the form%
\begin{equation}
a_{n}(x)=\sum_{j=0}^{n}\left(  -1\right)  ^{j}\left(  \varphi_{j}%
(x)-x^{j}\right)  \sum_{k=j}^{n}\left(  -1\right)  ^{k}\frac{n!}%
{(n-k)!k!(k-j)!j!}x^{k-j}. \label{a_n}%
\end{equation}

\end{proposition}

\begin{proof}
Denote by $l_{k,n}:=\frac{\left(  -1\right)  ^{k}}{k!}\binom{n}{k}$ the
coefficient at $x^{k}$ of the Laguerre polynomial $L_{n}(x)$.

Consider the integral
\begin{align*}
\int_{-\infty}^{x}\mathbf{A}(x,y)L_{n}(x-y)dy  &  =\sum_{n=0}^{\infty}%
a_{m}(x)\int_{-\infty}^{x}L_{m}(x-y)L_{n}(x-y)e^{-(x-y)}dy\\
&  =\sum_{n=0}^{\infty}a_{m}(x)\int_{0}^{\infty}L_{m}(t)L_{n}(t)e^{-(t)}%
dy=a_{n}(x).
\end{align*}
Hence
\begin{align*}
a_{n}(x)  &  =\int_{-\infty}^{x}\mathbf{A}(x,y)L_{n}(x-y)dy=\sum_{k=0}%
^{n}l_{k,n}\int_{-\infty}^{x}\mathbf{A}(x,y)\left(  x-y\right)  ^{k}dy\\
&  =\sum_{k=0}^{n}l_{k,n}\sum_{j=0}^{k}\left(  -1\right)  ^{j}\binom{k}%
{j}x^{k-j}\int_{-\infty}^{x}\mathbf{A}(x,y)y^{j}dy.
\end{align*}
From (\ref{Axk}) we have that $\int_{-\infty}^{x}\mathbf{A}(x,y)y^{j}%
dy=\varphi_{j}(x)-x^{j}$, and hence%
\[
a_{n}(x)=\sum_{k=0}^{n}l_{k,n}\sum_{j=0}^{k}\left(  -1\right)  ^{j}\binom
{k}{j}x^{k-j}\left(  \varphi_{j}(x)-x^{j}\right)
\]
from where (\ref{a_n}) is obtained by changing the order of summation.
\end{proof}

\begin{remark}
Although formula (\ref{a_n}) offers an explicit expression for the
coefficients, it is not the best alternative for practical computation.
Another representation for $a_{n}$ suited better for computational purposes is
derived in Section 4.
\end{remark}

\section{A representation for solutions of (\ref{Schr})}

\begin{theorem}
\label{Th representation of u}The solution $u(\omega,x)$ of (\ref{Schr})
satisfying the initial conditions (\ref{init u}) has the form%
\begin{equation}
u(\omega,x)=e^{-i\omega x}\left(  1+\sum_{n=0}^{\infty}\left(  -1\right)
^{n}a_{n}(x)\frac{\left(  i\omega\right)  ^{n}}{\left(  1-i\omega\right)
^{n+1}}\right)  . \label{u}%
\end{equation}

The following estimate is valid for any $\omega\in\mathbb{R}$,%
\begin{equation}
\left\vert u(\omega,x)-u_{N}(\omega,x)\right\vert \leq\varepsilon_{N}\left(
x\right)  , \label{estimate omega}%
\end{equation}
where
\[
u_{N}(\omega,x):=e^{-i\omega x}\left(  1+\sum_{n=0}^{N}\left(  -1\right)
^{n}a_{n}(x)\frac{\left(  i\omega\right)  ^{n}}{\left(  1-i\omega\right)
^{n+1}}\right)  ,
\]
and $\varepsilon_{N}\left(  x\right)  $ is a nonnegative function independent
of $\omega$ and such that $\varepsilon_{N}\left(  x\right)  \rightarrow0$ for
all $x\in\left[  0,d\right]  $ when $N\rightarrow\infty$.
\end{theorem}

\begin{proof}
Consider (\ref{u=Ae}). Substitution of (\ref{A}) into it gives us the
equality
\begin{align*}
u(\omega,x)  &  =e^{-i\omega x}+\sum_{n=0}^{\infty}a_{n}(x)\int_{-\infty}%
^{x}L_{n}(x-y)e^{-(x-y)}e^{-i\omega y}dy\\
&  =e^{-i\omega x}\left(  1+\sum_{n=0}^{\infty}a_{n}(x)\int_{0}^{\infty}%
L_{n}(t)e^{-(1-i\omega)t}dt\right)  .
\end{align*}
For the last integral the following equality holds \cite[7.414 (2)]{GR}%
\[
\int_{0}^{\infty}L_{n}(t)e^{-(1-i\omega)t}dt=\frac{\left(  -1\right)
^{n}\left(  i\omega\right)  ^{n}}{\left(  1-i\omega\right)  ^{n+1}},
\]
from where we obtain (\ref{u}).

To prove (\ref{estimate omega}) consider the difference%
\begin{align*}
\left\vert u(\omega,x)-u_{N}(\omega,x)\right\vert  &  =\left\vert \int
_{0}^{\infty}e^{-t}\left(  \mathbf{a}\left(  x,t\right)  -\mathbf{a}%
_{N}\left(  x,t\right)  \right)  e^{-i\omega(x-t)}dt\right\vert \\
&  =\left\vert \left\langle \mathbf{a}\left(  x,t\right)  -\mathbf{a}%
_{N}\left(  x,t\right)  ,e^{-i\omega(x-t)}\right\rangle \right\vert ,
\end{align*}
where $\mathbf{a}_{N}\left(  x,t\right)  :=\sum_{n=0}^{N}a_{n}(x)L_{n}(t)$.
Application of the Cauchy--Bunyakovsky--Schwarz inequality leads to the
inequality%
\[
\left\vert u(\omega,x)-u_{N}(\omega,x)\right\vert \leq\left\Vert
\mathbf{a}\left(  x,t\right)  -\mathbf{a}_{N}\left(  x,t\right)  \right\Vert
\left\Vert e^{-i\omega(x-t)}\right\Vert
\]
where obviously $\left\Vert e^{-i\omega(x-t)}\right\Vert =1$ for $\omega
\in\mathbb{R}$. Now denoting
\[
\varepsilon_{N}\left(  x\right)  :=\left\Vert \mathbf{a}\left(  x,t\right)
-\mathbf{a}_{N}\left(  x,t\right)  \right\Vert =\left(  \int_{0}^{\infty
}e^{-t}\left\vert \mathbf{a}\left(  x,t\right)  -\mathbf{a}_{N}\left(
x,t\right)  \right\vert ^{2}dt\right)  ^{\frac{1}{2}}%
\]
we obtain (\ref{estimate omega}).
\end{proof}

\begin{remark}
\label{Rem complex omega}The series (\ref{u}) can be considered also in the
case $\omega\in\mathbb{C}$, though not for all values of $\omega$ such
attractive estimates as (\ref{estimate omega}) are possible. Indeed, first we
note that the function $e^{i\omega t}\in L_{2}\left(  0,\infty;e^{-t}\right)
$ if only $\operatorname{Im}\omega>-1/2$ and in this case $\left\Vert
e^{i\omega t}\right\Vert =1/\sqrt{1+2\operatorname{Im}\omega}$. Thus, applying
the reasoning from the proof of Theorem \ref{Th representation of u} we obtain%
\[
\left\vert u(\omega,x)-u_{N}(\omega,x)\right\vert \leq\frac{\varepsilon
_{N}\left(  x\right)  e^{\operatorname{Im}\omega x}}{\sqrt
{1+2\operatorname{Im}\omega}},\qquad\text{when }\operatorname{Im}\omega>-1/2.
\]
This estimate is independent of $\operatorname{Re}\omega$ and is especially
attractive when $-1/2<\operatorname{Im}\omega\leq0$. We discuss its
applications in Section \ref{Sect Numerical}.
\end{remark}

\section{A recurrent procedure for computing the coefficients $a_{n}$}

In order to obtain another way to compute $a_{n}$ we substitute (\ref{u}) into
equation (\ref{Schr}). Let us stress that it is done formally without
discussing the possibility of differentiation of the series (\ref{u}) with
respect to $x$, and only with the aim to obtain the recurrent formulas which
then can be easily checked directly from (\ref{a_n}).

Differentiating twice (\ref{u}) we obtain the equality%
\[
u^{\prime\prime}(\omega,x)+\omega^{2}u(\omega,x)=e^{-i\omega x}\left(
\sum_{n=0}^{\infty}a_{n}^{\prime\prime}(x)\frac{\left(  -i\omega\right)  ^{n}%
}{\left(  1-i\omega\right)  ^{n+1}}-2i\omega\sum_{n=0}^{\infty}a_{n}^{\prime
}(x)\frac{\left(  -i\omega\right)  ^{n}}{\left(  1-i\omega\right)  ^{n+1}%
}\right)  .
\]
Using (\ref{Schr}) on the left-hand side we have that%
\[
q(x)\left(  1+\sum_{n=0}^{\infty}a_{n}(x)\frac{\left(  -i\omega\right)  ^{n}%
}{\left(  1-i\omega\right)  ^{n+1}}\right)  =\sum_{n=0}^{\infty}\left(
a_{n}^{\prime\prime}(x)\frac{\left(  -i\omega\right)  ^{n}}{\left(
1-i\omega\right)  ^{n+1}}+2a_{n}^{\prime}(x)\frac{\left(  -i\omega\right)
^{n+1}}{\left(  1-i\omega\right)  ^{n+1}}\right)  .
\]
Denote $z:=\frac{i\omega}{1-i\omega}$. Then the last equality can be written
as follows%
\[
q(x)\left(  1+\frac{1+z}{z}\sum_{n=0}^{\infty}\left(  -1\right)  ^{n}%
a_{n}(x)z^{n+1}\right)  =\sum_{n=0}^{\infty}\left(  \left(  -1\right)
^{n}a_{n}^{\prime\prime}(x)\frac{1+z}{z}z^{n+1}+\left(  -1\right)
^{n+1}2a_{n}^{\prime}(x)z^{n+1}\right)
\]
or%
\[
q(x)\left(  1+\left(  1+z\right)  \sum_{n=0}^{\infty}\left(  -1\right)
^{n}a_{n}(x)z^{n}\right)  =\sum_{n=0}^{\infty}\left(  \left(  -1\right)
^{n}a_{n}^{\prime\prime}(x)\left(  1+z\right)  z^{n}+\left(  -1\right)
^{n+1}2a_{n}^{\prime}(x)z^{n+1}\right)  .
\]
Equating terms corresponding to equal powers of $z$ we obtain the equalities%
\begin{equation}
\mathbf{L}a_{n}=\mathbf{L}a_{n-1}-2a_{n-1}^{\prime} \label{La_n}%
\end{equation}
where $\mathbf{L}:=\frac{d^{2}}{dx^{2}}-q$. Moreover, from (\ref{a_n}) it
follows that
\begin{equation}
a_{0}=f_{0}-1, \label{a0}%
\end{equation}
and
\begin{equation}
a_{n}(0)=a_{n}^{\prime}(0)=0 \label{init a_n}%
\end{equation}
for all $n=0,1,\ldots$. From here a simple recurrent procedure for computing
the coefficients $a_{n}$ can be proposed which is formulated as the following statement.

\begin{proposition}
The coefficients $a_{n}$ in (\ref{A}) and (\ref{u}) can be calculated by means
of the following recurrent integration procedure%
\begin{equation}
a_{n}(x)=a_{n-1}(x)-2f(x)\int_{0}^{x}\frac{a_{n-1}(s)ds}{f(s)}+2f(x)\int
_{0}^{x}\frac{1}{f^{2}(t)}\int_{0}^{t}f^{\prime}(s)a_{n-1}(s)ds
\label{a_n recurrent}%
\end{equation}
with $f$ being any solution of (\ref{SchrHom}), nonvanishing on $\left[
0,d\right]  $, see Remark \ref{Rem On particular solutions}, and $a_{0}$
defined by (\ref{a0}).

\begin{proof}
One of the possibilities to write down a solution of the inhomogeneous
equation $\mathbf{L}a_{n}=g$ satisfying (\ref{init a_n}) is using the formula
\[
a_{n}(x)=f(x)\int_{0}^{x}\frac{1}{f^{2}(t)}\int_{0}^{t}f(s)g(s)ds.
\]
Its application to (\ref{La_n}) gives us the following equality
\[
a_{n}(x)=a_{n-1}(x)-2f(x)\int_{0}^{x}\frac{1}{f^{2}(t)}\int_{0}^{t}%
f(s)a_{n-1}^{\prime}(s)ds
\]
which after an integration by parts can be written in the form
(\ref{a_n recurrent}).
\end{proof}
\end{proposition}

\section{Numerical illustrations\label{Sect Numerical}}

\begin{example}
\label{Ex 1}Consider $q\equiv1$. This elementary example is sufficient to show
main features of the proposed representation for solutions of (\ref{Schr}).
Note that although the solutions can be written in a closed form the
coefficients $a_{n}$ are rather nontrivial since $f_{0}(x)=\cosh x$ and the
recurrent integrals in (\ref{phik}) or in (\ref{a_n recurrent}) can be
calculated explicitly only for few first coefficients. Thus, formula
(\ref{a_n recurrent}) was implemented numerically, and first we notice that it
allows one to compute, if necessary, dozens of thousands of the coefficients
$a_{n}$. We illustrate this observation by Fig. 1 and Fig. 2 on which the
difference $\max_{x\in\left[  0,1\right]  }\left\vert \frac{1}{2}\int_{0}%
^{x}q(y)dy-\sum_{n=0}^{N}a_{n}(x)\right\vert $ is presented for $N$ from zero
to ten thousands and from ten thousands one to one hundred thousands
respectively. Both graphs show that the sum $\sum_{n=0}^{N}a_{n}(x)$ with the
computed coefficients $a_{n}$ tends to $\frac{1}{2}\int_{0}^{x}q(y)dy$
according to (\ref{Sum a_n}) still for large $N$. We observe a quite slow
decay of the coefficients $a_{n}$. On Fig. 3 the magnitude $\max_{x\in\left[
0,1\right]  }\left\vert a_{n}(x)\right\vert $ is presented for $n=0,1,\ldots
100$ while on Fig. 4 the same magnitude is presented for $n=101,\ldots1000$. 

Fig. 5 shows the absolute error of the approximate solution for $N=100$ while
Fig. 6 shows the same magnitude but for $N=10^{4}$. Finally, we illustrate the
uniformity of the approximation with respect to $\omega\in\mathbb{R}$
prescribed by (\ref{estimate omega}). On Fig. 7 we show the absolute error of
the approximate solution for $N=100$ on the interval $\omega\in\left[
-1000,1000\right]  $. It can be observed that the maximum of the error is
achieved relatively near the origin, in this case when $\left\vert
\omega\right\vert \approx10$, and the error decays for large values of
$\omega$, for $\omega=\pm1000$ it is of order $10^{-6}$.

Fig.8 shows the absolute error of the approximate solution for $N=10^{4}$ on
the same interval $\omega\in\left[  -1000,1000\right]  $. Now the maximum of
the absolute error though two orders better than in the previous test is
achieved when $\left\vert \omega\right\vert \approx78$.In the next experiment
we chose $\omega=-i/4$, see Remark \ref{Rem complex omega}. Already with
$N=30$ the value of both the absolute and the relative errors of the solution
was of order $10^{-16}$. That confirms the estimate from Remark
\ref{Rem complex omega} and shows that the series (\ref{u}) converges
especially fast when $-1/2<\operatorname{Im}\omega<0$ and $\operatorname{Re}%
\omega=0$. This can be used for computing solutions for $\omega$ belonging to
other intervals by considering $\omega^{2}=\omega_{0}^{2}-\Lambda$ where
$-1/2<\operatorname{Im}\omega_{0}<0$ and $\operatorname{Re}\omega_{0}=0$, and
the equation is written in the form $-u^{\prime\prime}+\left(  q(x)+\Lambda
\right)  u=\omega_{0}^{2}u$. For example, we tested this simple idea
considering $\omega_{0}=-i/4$ and $\Lambda=100$, computing thus a solution for
Example \ref{Ex 1} with $\omega^{2}=-100.0625$. Again already with $N=30$ the
value of both the absolute and the relative errors of the solution was of
order $10^{-16}$.

\bigskip
\end{example}


\begin{thebibliography}{99}                                                                                               %


\bibitem {Faddeev}L. D. Faddeev, \emph{The inverse problem in the quantum
theory of scattering. II,} J. Soviet Math., 5:3 (1976), 334--396.

\bibitem {GR}I. Gradshteyn and I. Ryzhik, \textit{Table of integrals, series,
and products}, Academic Press, 1980.

\bibitem {KKRosu}K. V. Khmelnytskaya, V. V. Kravchenko and H. C. Rosu,
\emph{Eigenvalue problems, spectral parameter power series, and modern
applications}, Math. Methods Appl. Sci. 38 (2015), 1945--1969.

\bibitem {KrCV08}V. V. Kravchenko, \emph{A representation for solutions of the
Sturm-Liouville equation}, Complex Variables and Elliptic Equations, 53
(2008), 775--789.

\bibitem {KMoT}V. V. Kravchenko, S. Morelos and S. Tremblay \emph{Complete
systems of recursive integrals and Taylor series for solutions of
Sturm-Liouville equations.} Math. Methods Appl. Sci. 35 (2012), 704--715.

\bibitem {KNT}V. V. Kravchenko, L. J. Navarro, S. M. Torba
\emph{Representation of solutions to the one-dimensional Schr\"{o}dinger
equation in terms of Neumann series of Bessel functions.} Submitted, available
at arXiv.

\bibitem {KrPorter2010}V. V. Kravchenko and R. M. Porter, \emph{Spectral
parameter power series for Sturm-Liouville problems}, Math. Methods Appl.
Sci., 33 (2010), 459--468.

\bibitem {KT Birkhauser 2013}V. V. Kravchenko, S. M. Torba
\emph{Transmutations and spectral parameter power series in eigenvalue
problems.} In: Operator Theory: Advances and Applications, 2013, v. 228
Operator Theory, Pseudo-Differential Equations, and Mathematical Physics,
209-238, Springer: Basel.

\bibitem {KT PQR}V. V. Kravchenko, S. M. Torba, \emph{A Neumann series of
Bessel functions representation for solutions of Sturm-Liouville equations,
}submitted, available from arXiv.org.

\bibitem {KTC}V. V. Kravchenko, S. M. Torba and R. Castillo-P\'{e}rez, \emph{A
Neumann series of Bessel functions representation for solutions of perturbed
Bessel equations,} submitted, available at arXiv.

\bibitem {Marchenko}V. A. Marchenko, \emph{Sturm-Liouville operators and
applications: revised edition}, AMS Chelsea Publishing, 2011.

\bibitem {Marchenko52}V. A. Marchenko, \emph{Some questions on one-dimensional
linear second order differential operators}, Transactions of Moscow Math.
Soc., 1 (1952), 327--420.
\end{thebibliography}
\end{document}